\documentclass[a4paper,10pt]{article}
\usepackage[T1]{fontenc}
\usepackage[utf8]{inputenc}
\usepackage{graphicx}
\usepackage[dvipsnames]{xcolor}
\usepackage{url}
\usepackage{amsmath,amssymb,amsthm}
\usepackage{mathtools}
\usepackage[labelformat=simple]{subcaption}
\usepackage{xspace}
\usepackage[hidelinks]{hyperref}
\usepackage[noabbrev,capitalize]{cleveref}
\usepackage[margin=3cm]{geometry}

\renewcommand{\dh}{d_H}
\newcommand{\dhd}{d_{\vec{H}}}
\newcommand{\dd}{\vec{d}}
\newcommand{\dilate}[1]{#1^\oplus}
\newcommand{\eps}{\varepsilon}

\newcommand{\coll}{\mathcal M}

\DeclareMathOperator*{\argmin}{argmin}
\newcommand{\R}{\ensuremath{\mathbb R}\xspace}
\newcommand{\seg}{\ensuremath{\textnormal{seg}}}
\newcommand{\magicValue}{magic value\xspace}

\theoremstyle{plain}
\newtheorem{theorem}{Theorem}
\newtheorem{lemma}[theorem]{Lemma}

\title{Between Shapes, Using the Hausdorff Distance%
\footnote{Research on the topic of this paper was initiated at the 4th Workshop on Applied Geometric Algorithms (AGA 2018) in Langbroek, The Netherlands, supported by the Netherlands Organisation for Scientific Research (NWO) under project no. 639.023.208. The second author is supported by the NWO Veni grant EAGER. The first and fifth authors are supported by the NWO TOP grant no.~612.001.651.}
}
\author{Marc van Kreveld \and Tillmann Miltzow \and Tim Ophelders \and Willem Sonke \and Jordi L. Vermeulen}
\date{}

\graphicspath{{./figures/}}

\begin{document}

\maketitle

\begin{abstract}
Given two
shapes
$A$ and~$B$ in the plane with Hausdorff distance~$1$, is there
a shape~$S$ with Hausdorff distance $1/2$ to and from $A$ and~$B$?
The answer is always yes, and depending on convexity of~$A$ and/or~$B$,
$S$ may be convex, connected, or disconnected.
We show that our result can be generalised to give an interpolated shape
between \(A\) and \(B\) for any interpolation variable \(\alpha\) between 0 and
1, and prove that the resulting morph has a bounded rate of change with respect
to \(\alpha\).
Finally, we explore a generalization of the concept of a Hausdorff middle to
more than two input sets. We show how to approximate or compute this middle
shape, and that the properties relating to the connectedness of the Hausdorff
middle extend from the case with two input sets. We also give bounds on the
Hausdorff distance between the middle set and the input.
\end{abstract}

\section{Introduction}

\begin{figure}[t]
    \centering
    \includegraphics{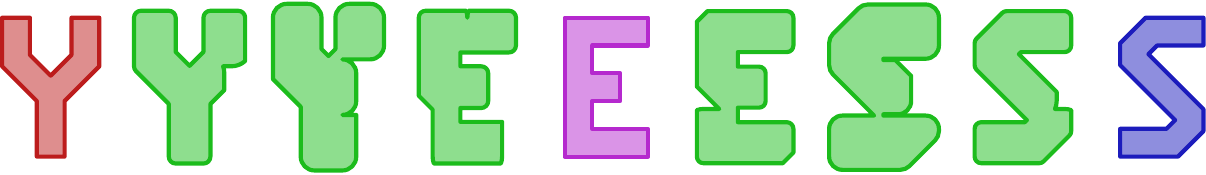}
    \caption{Hausdorff morphs between three shapes.}
    \label{fig:morphs-intro}
\end{figure}

For two sets \(A\) and~\(B\) in~\(\mathbb{R}^2\), we define the \emph{directed Hausdorff distance} as
\begin{equation*}
    \dhd(A, B) \coloneqq \adjustlimits\sup_{a \in A} \inf_{b \in B} d(a, b) \text{,}
\end{equation*}
where \(d\) denotes the Euclidean distance. The \emph{undirected Hausdorff distance} is defined as
\begin{equation*}
    \dh(A, B) \coloneqq \max(\dhd(A, B), \dhd(B, A)) \text{.}
\end{equation*}
If $A$ and $B$ are closed sets then
\(\dh(A, B) = r\)
is equivalent to saying that \(r\) is the smallest value such that \(A \subseteq B \oplus D_r\)
and \(B \subseteq A \oplus D_r\), where \(\oplus\) denotes the Minkowski sum,
and \(D_r\) is a disk of radius \(r\) centered at the origin. Recall that the
Minkowski sum of sets \(A\) and~\(B\) is the set \(\{ a + b~|~a \in A,~b \in B\}\).
In this paper we consider only closed sets, and therefore we can freely use this containment property.

The Hausdorff distance has been widely used in computer vision~\cite{dubuisson1994modified} and computer graphics~\cite{aspert2002mesh,cignoni1998metro} for tasks such as template matching, and error computation between a model and its simplification. At the same time,
the Hausdorff distance is a classic mathematical concept.
Our research motivation is to study this profound concept from a new perspective.
Algorithms to compute the Hausdorff distance between two given sets are available for many types
of sets, such as points, line segments, polylines, polygons, and
simplices in $k$-dimensional Euclidean space~\cite{alt95,alt03,atallah83}.
However, the question whether a polynomial-time algorithm
exists to compute the Hausdorff distance between general semialgebraic sets remains open~\cite{dobbins18}.

In this paper, we consider the natural problem of finding a set that lies ``between'' two or more input sets, in a Hausdorff sense.
In Section~\ref{sec:2sets} we investigate
the Hausdorff middle of sets $A$ and~$B$; this is a set that has minimum undirected Hausdorff distance to $A$ and~$B$. Differently put, it minimizes the maximum of four directed Hausdorff distances.
We show that when the Hausdorff distance between $A$ and~$B$ is assumed to be $1$, there is always a Hausdorff middle that has Hausdorff distance $1/2$ to $A$ and~$B$, and this is the best possible. We relate the convexity of $A$ and/or $B$ to the convexity and connectedness of the Hausdorff middle, and study its combinatorial complexity.

We actually treat the middle more generally, by defining
a class of sets that smoothly interpolate between $A$ and~$B$, giving a morph between them. Figure~\ref{fig:morphs-intro} shows two examples of such morphs. We prove that for two given intermediate shapes in the morph, the difference between the interpolation parameters bounds the Hausdorff distance between the shapes.

Algorithms for morphing, sometimes called \emph{shape interpolation}, have been widely studied. A classical application is the reconstruction of a 3D object from 2D slices, a common problem in medical imaging. Many algorithms that solve this problem exist, based on straight skeletons~\cite{barequet04,barequet08}, curve matching and triangulations~\cite{barequet96}, and Delaunay triangulations~\cite{boissonnat88}. When considering more abstract applications, a typical approach is to first transform each input shape into a cannonical form, and then morph between those. Alt and Guibas~\cite{alt00} give an overview of this approach. Finally, work has been done to ensure the interpolation of two simple polygons is itself a simple polygon~\cite{gotsman01}.

A common thread in all these algorithms is that they are based on computing some kind of correspondence between features of the input shapes, either by explicitly matching parts of the boundary, or by computing some geometrical structure (like a Voronoi diagram or a straight skeleton). In addition, most of these morphing algorithms interpolate only the boundary of the input shapes, and keep all intermediate shapes polygonal. Our approach does not require any correspondence between features of the input to be calculated. However, our approach is unusual in the sense that the intermediate shapes when morphing between e.g.\ two polygons are not necessarily polygons themselves.

In \cref{sec:3sets} we extend the results of \cref{sec:2sets} to Hausdorff middles of more than two sets and generalize several results.
We assume that the maximum Hausdorff distance over all pairs of input sets is $1$ and
examine the smallest Hausdorff distance for a middle set.
That is, given sets $\coll = \{A_1,\ldots,A_k\} $, we are interested in the value $\alpha(\coll) = \min_{S} \max_{i = 1,\ldots,k} \dh(A_i,S)$.
This value $\alpha(\coll)$ is no longer $1/2$, but depends on the input.
For convex sets, we show that a value $\approx 0.608$ can always be achieved and is sometimes necessary, whereas for non-convex sets a value of $1$ may be required.
For a given set of polygons with total combinatorial complexity $n$, we show that $\alpha(\coll)$ and the Hausdorff middle can be computed in $O(n^6)$ time, and, for any constant $\eps>0$, $(1+\eps)$-approximated in $O(n^2\log^2n\log 1/\eps)$ time.
We note that other interpolation methods between two shapes do not have a natural generalization to a middle of three or more shapes.

Our proofs use three types of arguments.
First, many of our arguments rely on simple manipulations of the formal definition of the Hausdorff distance.
The second type of argument is of a topological nature.
Using continuity and connectivity, we infer related properties to the output,
by constructing topological structures or conclude that they cannot exist.
The third type of argument uses $2$-dimensional Euclidean geometry directly.
We construct features, like vertices, edges and circular arcs, and argue about their
existence, and give distance bounds.
These arguments are often intricate
and do not generalize.
They are of particular value, as the $2$-dimensional Euclidean plane
is often the most interesting case in computational geometry.

\begin{figure}
    \centering
    \includegraphics[page=2]{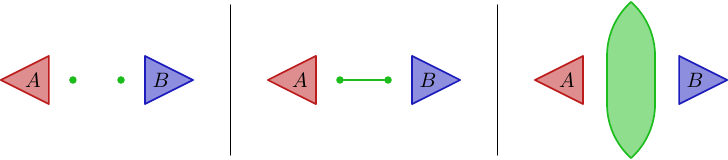}
    \caption{Three possible Hausdorff middles of $A$ and~$B$: two points, a line segment, and $S_{1/2}$.}
    \label{fig:middle-example}
\end{figure}

\section{The Hausdorff middle of two sets}
\label{sec:2sets}

Consider two compact sets \(A\) and~\(B\) in~\(\mathbb{R}^2\); we are interested in computing a \emph{Hausdorff middle}: a set~\(C\) that minimizes the maximum of the undirected Hausdorff distances to \(A\) and~\(B\). That is,
\begin{equation*}
    C \in \argmin_{C'}\max(\dh(A, C'), \dh(B, C')) \text{.}
\end{equation*}
Note that there may be many such sets that minimize the Hausdorff distance; see \cref{fig:middle-example} for a few examples. It might seem intuitive to restrict \(C\) to be the minimal set that achieves this distance, but such a set is not necessarily unique, and the common intersection of all minimal sets is not a solution itself (see Figure~\ref{fig:minimal-sets}).
However, the maximal set is unique. Let \(\dh(A, B) = 1\). Then
\begin{equation*}
    S(A, B) \coloneqq (A \oplus D_{1/2}) \cap (B \oplus D_{1/2})
\end{equation*}
is the unique maximal set with Hausdorff distance \(1/2\) to \(A\) and \(B\) (we prove this below in \cref{lem:maximal}; see the right of \cref{fig:middle-example} for an example of what \(S\) looks like). Note that in the rest of the paper we omit the arguments and simply write \(S\), as the arguments are always clear from context. We want to show that \(\dh(A, S)\leq 1/2\) and \(\dh(B, S)\leq 1/2\).
In fact, we can prove a more general statement.

We define
\begin{equation*}
    S_\alpha(A, B) \coloneqq (A \oplus D_\alpha) \cap (B \oplus D_{1 - \alpha})
\end{equation*}
for \(\alpha \in [0, 1]\), and we use $\seg(a,b)$ to denote the line segment connecting points $a$ and $b$.

\begin{figure}
    \centering
    \includegraphics[page=3]{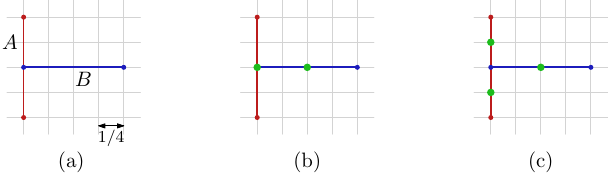}
    \caption{Two different minimal sets achieving minimal Hausdorff distance to $A$ and~$B$. Both the two green dots in Figure (b) and the three green dots in Figure (c) minimise the Hausdorff distance to \(A\) and \(B\).}
    \label{fig:minimal-sets}
\end{figure}

\begin{figure}
    \centering
    \includegraphics{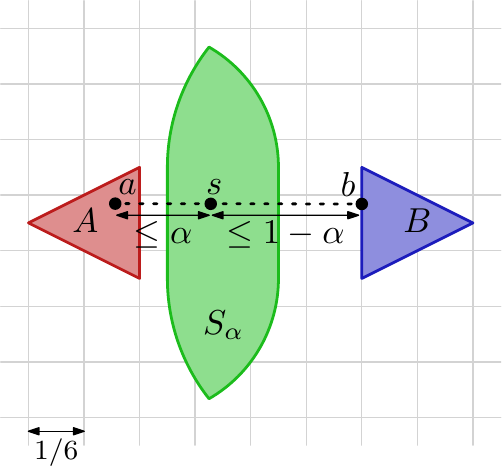}
    \caption{An arbitrary point \(a \in A\) with its closest point \(b\) on \(B\). The point \(s\) has distance at most \(\alpha\) to \(a\), and distance at most \(1 - \alpha\) to \(b\).}
    \label{fig:hausdorff-bound}
\end{figure}

\begin{theorem}
    Let \(A\) and \(B\) be two  compact sets in the plane with \(\dh(A, B) = 1\). Then \(\dh(A, S_\alpha) = \alpha\) and \(\dh(B, S_\alpha) = 1 - \alpha\).
\end{theorem}
\begin{proof}
    We first show that \(\dh(A, S_\alpha) \leq \alpha\).
    The proof for \(\dh(B, S_\alpha) \leq 1 - \alpha\) is analogous and therefore omitted.
    We will infer \(\dh(A, S_\alpha) \leq \alpha\) from \(\dhd(A, S_\alpha) \leq \alpha\) and \(\dhd(S_\alpha, A) \leq \alpha\); thereafter we will show equality.

    Consider any point \(a \in A\); by our assumption that \(\dh(A, B) = 1\), there is a point \(b \in B\) with \(d(a, b) \leq 1\); see \cref{fig:hausdorff-bound}. Now consider a point \(s \in \seg(a,b)\) with \(d(a, s) \leq \alpha\) and \(d(b, s) \leq 1 - \alpha\); clearly this point must be in \(S_\alpha\), as it is contained in both \(A\oplus D_\alpha\) and \(B\oplus D_{1-\alpha}\), and it has \(d(a, s) \leq \alpha\). As this works for every \(a \in A\), it holds that \(\dhd(A, S_\alpha) \leq \alpha\). The fact that \(\dhd(S_\alpha, A) \leq \alpha\) follows straightforwardly from \(S_\alpha\) being a subset of \(A\oplus D_{\alpha}\). Thus, \(\dh(A, S_\alpha) \leq \alpha\).

    To show equality, assume that the Hausdorff distance between \(A\) and \(B\) is realized by a point $\hat{a} \in A$ with closest point $\hat{b} \in B$, at distance~$1$.
    Consider the point \(\hat{s} \in \seg(\hat{a},\hat{b})\) with \(d(\hat{a}, \hat{s}) = \alpha\) and \(d(\hat{b}, \hat{s}) = 1 - \alpha\). As observed, $\hat{s}\in S_\alpha$. Since $\hat{s}$ is the closest point of $S_\alpha$ to $\hat{a}$, and $\hat{b}$ is the closest point of $B$ to $\hat{s}$, equality follows.
\end{proof}
\begin{lemma}\label{lem:maximal}
    \(S_\alpha\) is the maximal set that satisfies \(\dh(A, S_\alpha) = \alpha\) and \(\dh(B, S_\alpha) = 1 - \alpha\).
\end{lemma}
\begin{proof}
    Consider any set \(T\) for which we have \(\dhd(T, A) \leq \alpha\) and \(\dhd(T, B) \leq 1 - \alpha\). As \(A\oplus D_\alpha\) contains all points with distance at most \(\alpha\) to \(A\), we have that \(T \subseteq A\oplus D_\alpha\); similarly, we have that \(T \subseteq B\oplus D_{1-\alpha}\). By the definition of \(S_\alpha\), this implies that \(T \subseteq S_\alpha\). As this holds for any \(T\), we conclude that \(S_\alpha\) is maximal.
\end{proof}

\subsection{Properties of \texorpdfstring{\(\boldsymbol{S_\alpha}\)}{S-alpha}}\label{sec:two-sets-properties}
\begin{figure}
    \centering
    \includegraphics[page=2]{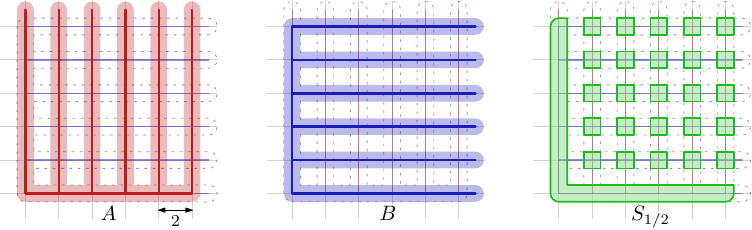}
    \caption{Sets $A$ and~$B$ for which $S_{1/2}$ is disconnected. The shaded areas around $A$ and~$B$ represent $A \oplus D_{1/2}$ and $B \oplus D_{1/2}$, respectively.}
    \label{fig:two-sets-quadratic-lower-bound}
\end{figure}

In this section, we study the convexity and connectedness of \(S_\alpha\). Recall that a set \(A \subseteq \R^2\) is convex if for any two points \(a, b \in A\), the segment \(\seg(a,b)\)
between them is completely contained in \(A\). Also, recall that a set $A\subset \R^2$ is connected if for any two points $a, b \in A$, there exists a continuous curve $c : [0,1] \rightarrow A$ such that $c(0) = a$ and $c(1) = b$. This type of connectedness is known as path-connectedness, but we use the term connected for simplicity. We observe the following properties:
\begin{enumerate}
    \item If \(A\) and \(B\) are convex, \(S_\alpha\) is convex;
    \item If \(A\) is convex and \(B\) is connected, \(S_\alpha\) is connected;
    \item For some connected sets \(A\) and \(B\), \(S_\alpha\) is disconnected.
\end{enumerate}
Property~1 is straightforward: the Minkowski sum of \(A\) and \(B\) with a disk is convex, and the intersection of convex objects is itself also convex. The example in \cref{fig:two-sets-quadratic-lower-bound} demonstrates Property~3; in fact, \emph{any} Hausdorff middle will be disconnected for those input sets.

The next lemma establishes Property 2.

\begin{lemma}\label{lem:one-non-convex-connected}
    Let \(A\) and \(B\) be two connected regions of the plane with Hausdorff distance $1$, and \(A\) convex. Then \(S_\alpha = (A \oplus D_\alpha) \cap (B \oplus D_{1 - \alpha})\) is connected for any \(\alpha \in [0, 1]\).
\end{lemma}
\begin{proof}
    See Figure~\ref{fig:convex-implies-connected} for an illustration.
    Because \(A\) is convex, there is a continuous map $\rho\colon B\to A$ that maps each point of $B$ to a closest point (within distance $1$) in~$A$.
    For $b\in B$, let $\rho_\alpha(b)= \alpha \rho(b) + (1-\alpha)b$. We have that $\rho_\alpha\colon B\to S_\alpha$ is also continuous.

    Now take any two points \(s\) and \(s'\) in \(S_\alpha\); respectively, they have
        points $b$ and $b'\in B$ within distance $1-\alpha$.
    The segments between $s$ and $\rho_\alpha(b)$ and between $s'$ and $\rho_\alpha(b')$ lie completely in $S_\alpha$.
    Take a continuous curve \(\pi\) from \(b\) to \(b'\) inside \(B\). The image of $\pi$ under $\rho_\alpha$ connects $\rho_\alpha(b)$ to $\rho_\alpha(b')$ within $S_\alpha$, so $s$ and $s'$ are connected inside $S_\alpha$.
\end{proof}

\begin{figure}
    \begin{minipage}[t]{.47\textwidth}
        \centering
        \includegraphics{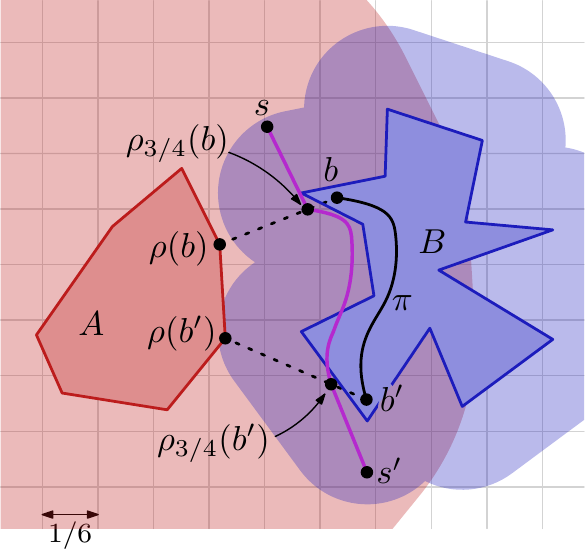}
        \caption{Illustration of the proof showing that $S_\alpha$ is connected if $A$ is convex (sketched for $\alpha = 3/4$). The shaded areas around $A$ and~$B$ represent $A \oplus D_{3/4}$ and $B \oplus D_{1/4}$, respectively, so that the doubly-shaded area is $S_{3/4}$.}
        \label{fig:convex-implies-connected}
    \end{minipage}%
    \hfill%
    \begin{minipage}[t]{.47\textwidth}
        \centering
        \includegraphics{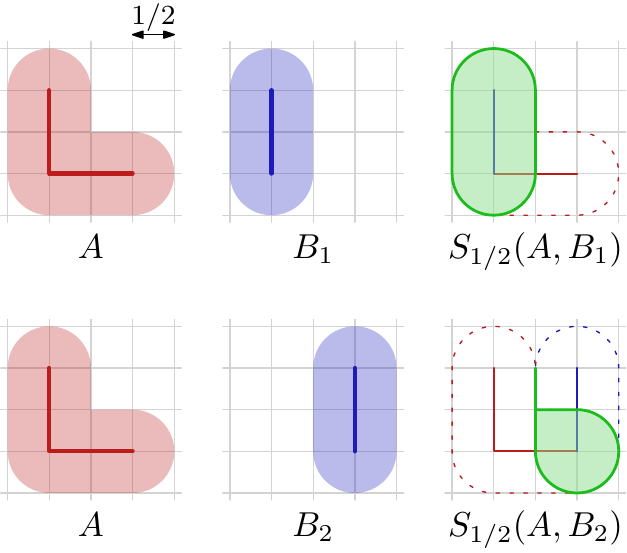}
        \caption{Although \(B_2\) is a translate of \(B_1\), the middle set between \(A\) and \(B_2\) is not a translate of the middle set between \(A\) and \(B_1\).}
        \label{fig:translate}
    \end{minipage}
\end{figure}

We note that $S_\alpha$ may contain holes. Furthermore, $S_\alpha$ is not shape invariant when $B$ is translated with respect to $A$. For example, let $A$ be the union of the left and bottom sides of a unit square and let $B_1$ and $B_2$ be the left and right sides of that same unit square.
Then $(A \oplus D_{1/2}) \cap (B_1 \oplus D_{1/2})$ is not a translate of
$(A \oplus D_{1/2}) \cap (B_2 \oplus D_{1/2})$. See \cref{fig:translate}; note that \(\dh(A, B_1) = \dh(A, B_2)\).

\subsection{Complexity of \texorpdfstring{\(\boldsymbol{S_\alpha}\)}{S-alpha}}
In this section, we describe the complexity of \(S_\alpha\) in terms of the number of vertices, line segments, and circular arcs on its boundary, for several types of polygonal input sets. Recall that \(\partial A\) denotes the boundary of set \(A\).

\begin{lemma}
    Let \(A\) be a convex polygon with $n$ vertices and \(B\) a simple polygon with $m$ vertices. Then \(\partial S_\alpha\) consists of \(O(n + m)\) vertices, line segments and circular arcs, and this bound is tight in the worst case.
\end{lemma}
\begin{proof}
    For brevity we let $\dilate{A}=A\oplus D_\alpha$
    and $\dilate{B}=B\oplus D_{1-\alpha}$.

    There is a trivial worst-case lower bound of \(\Omega(n + m)\) by taking \(\alpha = 0\) or \(\alpha = 1\), as \(S_0 = A\) and \(S_1 = B\).
    Note that if the boundaries of \(\dilate{A}\) and \(\dilate{B}\) would consist of only line segments, the upper bound is easy to show: \(\dilate{A}\) is convex, and its boundary can therefore intersect each segment of \(\partial \dilate{B}\) at most twice, making \(\partial S_\alpha\) consist of (parts of) segments from \(\partial\dilate{A}\) and \(\partial\dilate{B}\) and at most \(O(m)\) intersection points. The problem is that \(\partial\dilate{A}\) and \(\partial\dilate{B}\) also contain circular arcs, in which case \(\partial\dilate{A}\) may intersect an arc of \(\partial\dilate{B}\) \(\Omega(n)\) times.

    \begin{figure}
        \centering
        \includegraphics{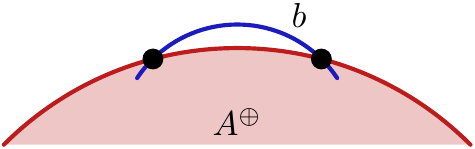}
        \caption{When \(\alpha \geq 1 - \alpha\), an arc \(b\) of \(\partial\dilate{B}\) (blue) can only intersect \(\partial\dilate{A}\) (red) twice.}
        \label{fig:two-intersections}
    \end{figure}

    To show an upper bound of \(O(n + m)\), we distinguish two cases. In the first case, we assume \(\alpha \geq 1 - \alpha\).
    Note that in this case, the circular arcs that are part of the boundary of \(\dilate{A}\) have a radius larger or equal to those of \(\dilate{B}\).
    Additionally, $\partial\dilate{A}$ is smooth and is an alternating sequence of circular arcs and segments, as $A$ is convex.
    In this case, we do in fact have that any line segment or circular arc \(b\) of \(\partial\dilate{B}\) can intersect \(\partial\dilate{A}\) at most twice. Consider two intersection points of \(b\) with \(\partial\dilate{A}\): as the curvature of \(\partial\dilate{A}\) is at most that of~\(b\), there can never be another intersection point between these two, or we would violate the convexity of \(\dilate{A}\). See \cref{fig:two-intersections} for an illustration of this case.

    For the second case, we assume \(\alpha < 1 - \alpha\). We charge all the intersections to the arcs and edges of \(\partial\dilate{A}\) and \(\partial\dilate{B}\). Each edge of \(\partial\dilate{B}\) can  intersect \(\partial\dilate{A}\) at most twice, as \(\dilate{A}\) is convex, so there can be at most \(O(m)\) such intersections. Similarly, for arcs of \(\partial\dilate{B}\) that intersect \(\partial\dilate{A}\) at most three times, there can be at most \(O(m)\) intersections in total. It remains to consider the arcs of \(\partial\dilate{B}\) that intersect \(\partial\dilate{A}\) more than three times.

    \begin{figure}
        \centering
        \includegraphics{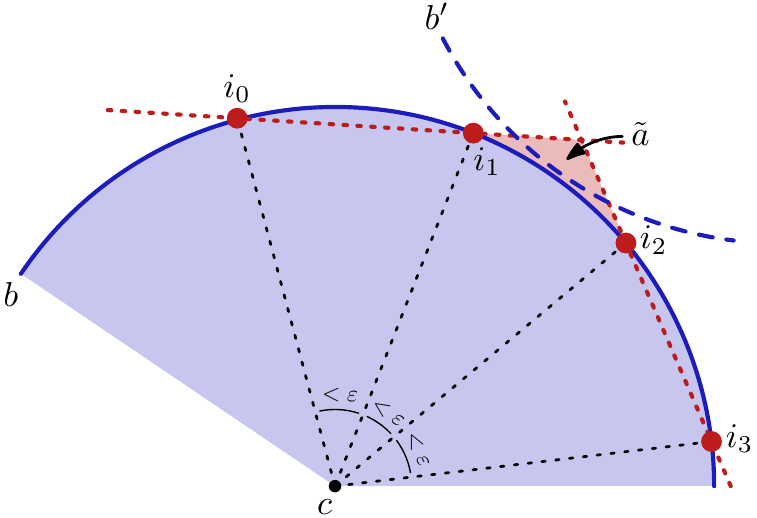}
        \caption{When \(\alpha < 1 - \alpha\), a single arc \(b\) of \(\partial\dilate{B}\), shown in blue, can have many intersections with \(\partial\dilate{A}\), but no other arc \(b'\), shown as a dashed blue arc, can have many intersections with the same part of \(\partial\dilate{A}\). The intersections of \(b\) with \(\partial\dilate{A}\) are shown in red.}
        \label{fig:intersection-bound}
    \end{figure}

    Let \(b\) be such an arc of \(\partial\dilate{B}\). Consider any quadruple of consecutive intersection points \(i_0, i_1, i_2, i_3\) with \(\partial\dilate{A}\) along \(b\), see \cref{fig:intersection-bound}, where the part of $\partial\dilate{A}$ between $i_1$ and $i_2$ that does not contain $i_0$ and $i_3$ is outside the disk supporting $b$. This part is denoted $\tilde{a}$; note that $\tilde{a}$ must contain at least one circular arc, denoted $a$.
    Notice that we consider all intersection points between $\partial \dilate{A}$ and $b$, except possibly for the first one or two and last one or two. These first and last ones can be charged to $b$, and this charge is at most four per arc $b$.
    Let $c$ be the center of the supporting disk of $b$. If any of the angles $\measuredangle i_0ci_1$,
    $\measuredangle i_1ci_2$, or
    $\measuredangle i_2ci_3$, is larger than $\varepsilon$ for some constant $\varepsilon>0$, we again charge the intersection points $i_1$ and $i_2$ to $b$, and we have less than $360/\varepsilon$ of such charges.
    So we now assume that all three angles are at most $\varepsilon$.
    We charge the intersection points $i_1$ and $i_2$ to $a$, the arc of a disk that appears on $\tilde{a}$.

    It remains to show that $a$ is charged at most once.
    We can limit the distance by which $\tilde{a}$ can protrude outside of \(b\): as \(\dilate{A}\) is convex, $\tilde{a}$ cannot cross the line through \(i_0\) and \(i_1\), nor the line through \(i_2\) and \(i_3\). This restricts  $\tilde{a}$ to the shaded area in \cref{fig:intersection-bound}.
    It is possible that $\tilde{a}$ intersects a different arc \(b'\) of \(\partial\dilate{B}\) in this shaded area. We observe that the disk that $b'$ is a part of cannot contain the intersection points $i_1$ and $i_2$, as otherwise those points would not be intersections of \(\partial\dilate{A}\) and \(\partial\dilate{B}\). Now \(b'\) can intersect \(\partial\dilate{A}\) at most twice, as more intersections would violate the convexity of \(\dilate{A}\). In particular, $b'$ cannot intersect \(\partial\dilate{A}\) four times,
    and hence $b'$ cannot charge intersections on it to $a$. We conclude that $a$ is charged only once.
    From this
    we conclude that there are at most \(O(n + m)\) intersection points in total, and that \(\partial S_\alpha\) therefore consists of at most \(O(n + m)\) vertices, line segments and circular arcs.
\end{proof}
\begin{lemma}
\label{lem:complexity}
    Let \(A\) and \(B\) be two simple polygons of \(n\) and \(m\) vertices, respectively. Then \(\partial S_\alpha\) consists of \(O(nm)\) vertices, line segments and circular arcs, and this bound is tight in the worst case.
\end{lemma}
\begin{proof}
    The worst-case lower bound of \(\Omega(nm)\) follows by taking \(A\) and \(B\) to be two rotated ``combs''; see \cref{fig:two-sets-quadratic-lower-bound}. For \(\alpha = 1/2\), \(S_\alpha\) consists of \(\Omega(nm)\) distinct components.
    The upper bound follows directly from the fact that \(A\oplus D_{\alpha}\) and \(B\oplus D_{1 - \alpha}\) have complexities $O(n)$ and $O(m)$, respectively. Each individual arc and edge on the boundaries of \(A\oplus D_{\alpha}\) and \(B\oplus D_{1 - \alpha}\) intersect at most a constant number of times, so we cannot have more than \(O(nm)\) intersection points.
\end{proof}

In fact, not just \(S_\alpha\), but \emph{any} Hausdorff middle has complexity $\Theta(nm)$ for the example in Figure~\ref{fig:two-sets-quadratic-lower-bound}. $S_\alpha$ is maximal, so the components cannot be connected without changing the Hausdorff distance to \(A\) or \(B\), and other middles must have at least some point in every component of $S_\alpha$ to achieve Hausdorff distance \(1/2\) to both \(A\) and \(B\).

\subsection{\texorpdfstring{\(\boldsymbol{S_\alpha}\)}{S-alpha} as a morph}
\label{sec:morphs}

\begin{figure}[p]
    \centering
    \includegraphics[page=1]{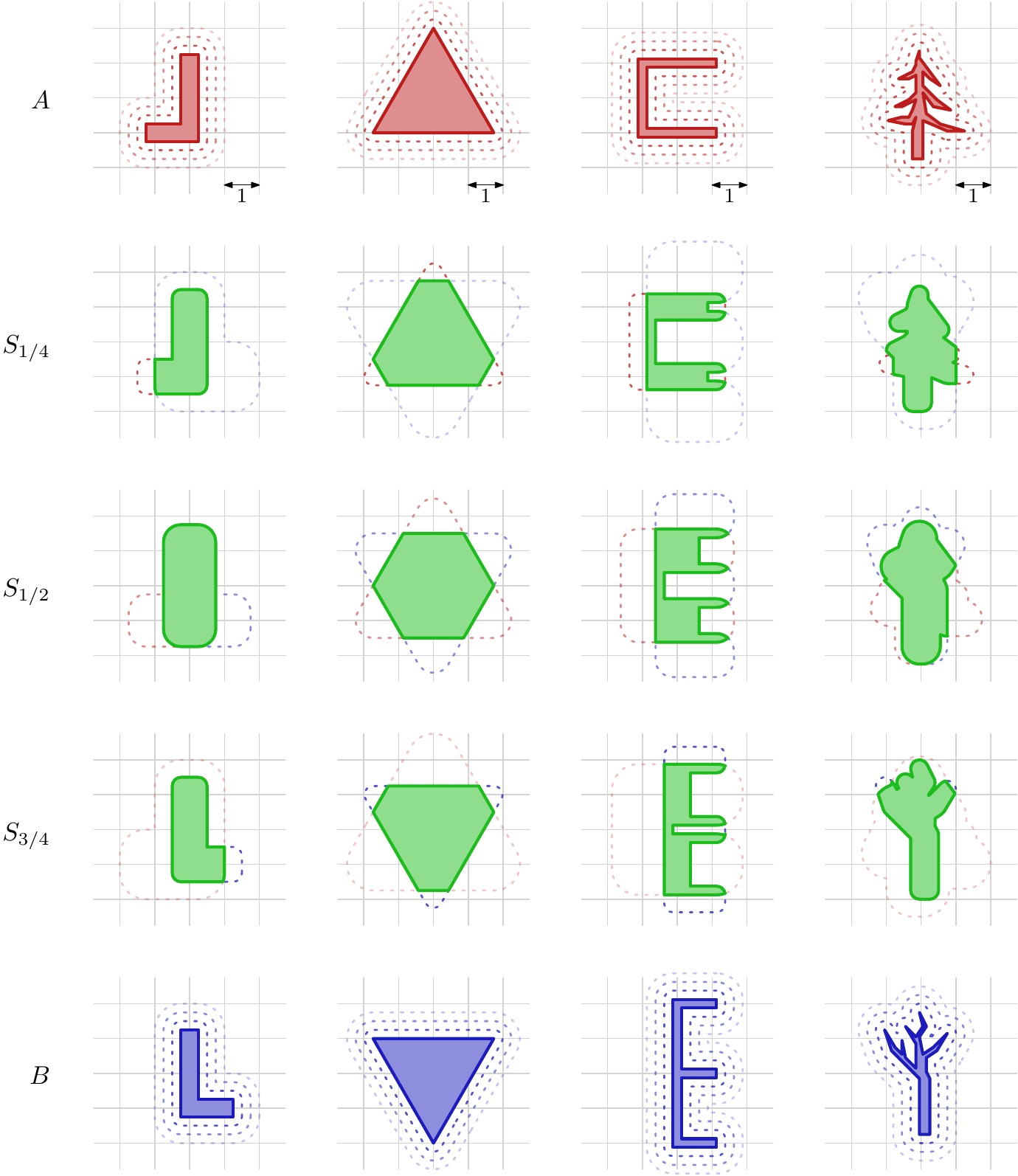}
    \caption{Some examples of morphs $S_\alpha$ between two shapes $A$ and~$B$.}
    \label{fig:morphs}
\end{figure}

By increasing $\alpha$ from $0$ to $1$, $S_\alpha$ morphs from $A=S_0$ into $B=S_1$. (Examples of such morphs are presented in \cref{fig:morphs-intro,fig:morphs}.) The following lemma shows that this morph has a bounded rate of change.

\begin{lemma}\label{lem:morph}
    Let \(S_\alpha\) and \(S_\beta\) be two intermediate shapes of \(A\) and \(B\) with \(d_H(A, B) = 1\) and \(\alpha \leq \beta\). Then \(\dh(S_\alpha, S_\beta) = \beta - \alpha\).
\end{lemma}
\begin{proof}
    We have $\dh(S_\alpha,S_\beta)\geq\beta-\alpha$ because, by the triangle inequality, $\dh(A,B)=1\leq\dh(A,S_\alpha)+\dh(S_\alpha,S_\beta)+\dh(S_\beta,B)\leq\alpha+\dh(S_\alpha,S_\beta)+1-\beta$.

    It remains to show that $\dh(S_\alpha,S_\beta)\leq\beta-\alpha$.
    We show that \(S_\beta \subseteq S_\alpha \oplus D_{\beta - \alpha}\); the proof that \(S_\alpha \subseteq S_\beta + D_{\beta-\alpha}\) is analogous. Let \(p\) be some point in \(S_\beta\). Then, by definition of \(S_\beta\), there exist some points \(a \in A\) and \(b \in B\) such that \(d(a, p) \leq \beta\) and \(d(b, p) \leq 1 - \beta\). Let \(\bar{p}\) be the point obtained by moving \(p\) in the direction of \(a\) by \(\beta - \alpha\). By the triangle inequality, we then have that \(d(a, \bar{p}) \leq \beta - (\beta - \alpha) = \alpha\) and \(d(b, \bar{p}) \leq (1 - \beta) + (\beta - \alpha) = 1 - \alpha\). This implies that \(\bar{p} \in S_\alpha\). As \(p\) was an arbitrary point in \(S_\beta\), and \(d(p, \bar{p}) \leq \beta - \alpha\), we have that \(S_\beta \subseteq S_\alpha \oplus D_{\beta - \alpha}\).
    So $\dh(S_\alpha,S_\beta)\leq\beta-\alpha$.
\end{proof}

The lemma implies that, even though the number of connected components of $S_\alpha$ can change when $\alpha$ changes, new components arise by splitting and never `out of nothing', and the number of components can only decrease through merging and not by disappearance.

The morph \(\langle S_\alpha\, |\, \alpha \in [0, 1]\rangle\) from $A$ to $B$ has a consistent submorph property, formalized below.

\begin{lemma}
If a morph from $A=S_0$ to $B=S_1$ contains a shape $C$, then the morph from $A$ to $C$ concatenated with the morph from $C$ to $B$ is the same as the morph from $A$ to $B$: they contain the same collection of shapes in between and in the same order.
\end{lemma}
\begin{proof}
    Let \(\alpha\) be the value such that \(S_\alpha(A, B) = C\). We define \(S'_\beta \coloneqq (A \oplus D_\beta) \cap (C \oplus D_{\alpha - \beta})\) for  \(\beta \in [0, \alpha]\), giving the morph from \(A\) to \(C\). We need to show that \(S_\beta(A, B) = S'_\beta(A, C)\). The case for the morph from \(C\) to \(B\) is analogous and therefore omitted.

    Let \(x\) be any point in \(S_\beta(A, B)\). By definition it has a distance of at most \(\beta\) to \(A\), and \cref{lem:morph} establishes that it has distance at most \(\alpha - \beta\) to \(C\). This implies that \(x \in S'_\beta(A, C)\). As this works for any point \(x\), we have that \(S_\beta(A, B) \subseteq S'_\beta(A, C)\). Now let \(x'\) be any point in \(S'_\beta(A, C)\). By definition it has distance at most \(\beta\) to \(A\), and distance at most \(\alpha - \beta\) to some point \(c \in C\). As \(\dh(B, C) = 1 - \alpha\), by the triangle inequality \(x'\) must have distance at most \((\alpha - \beta) + (1 - \alpha) = 1 - \beta\) to some point in \(B\). This shows \(x' \in S_\beta(A, B)\). As this works for any point \(x'\), we also have that \(S'_\beta(A, C) \subseteq S_\beta(A, B)\). We conclude that \(S_\beta(A, B) = S'_\beta(A, C)\).
\end{proof}

As a corollary of this lemma, $\{\alpha \in [0,1] \mid S_\alpha \textrm{ is convex}\}$ is a connected interval.

\subsection{The cost of connectedness}
For some applications, it might be necessary to insist that the middle shape is always connected. However, in the worst case, the cost of connecting all components of \(S_\alpha\) can be that the Hausdorff distance of the resulting shape to \(A\) and \(B\) becomes~$1$. See \cref{fig:cost-of-connectedness} for an example where this is the case. In fact, any connected shape has distance at least 1 for this example.

\begin{figure}
    \centering
    \includegraphics[page=2]{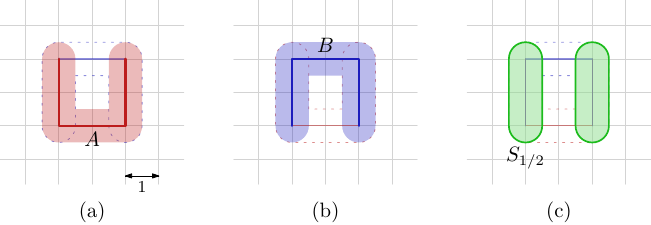}
    \caption{Figures (a) and (b) show the offsets of $A$, respectively $B$ with distance $1/2$. Figure (c) shows the resulting $S_{1/2}$ in green. Any
    connected shape must cross the vertical middle line or stay
    on one side of it. In both cases, the Hausdorff distance doubles.}
    \label{fig:cost-of-connectedness}
\end{figure}

\section{The Hausdorff middle of more than two sets}
\label{sec:3sets}

\begin{figure}
    \centering
     \includegraphics{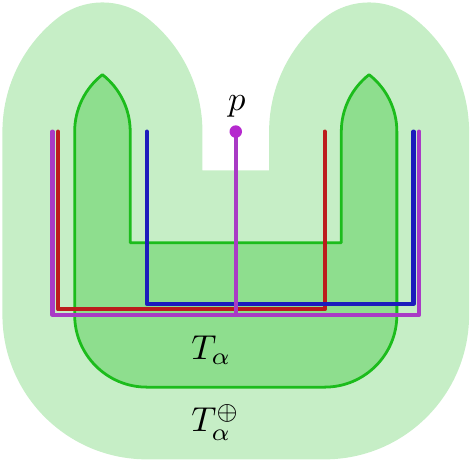}
    \caption{The pairwise Hausdorff distance in this construction is \(1\), and for any \(\alpha < 1\), \(\dilate{T_\alpha}\) does not contain point \(p\).}
    \label{fig:three-sets-connected-lower-bound}
\end{figure}

A natural question is whether the results from the previous section extend to more than two input shapes. There are several ways to formalise the notion of a Hausdorff middle between multiple shapes. Analogous to the case of two sets, we are interested in a middle shape that minimizes the maximum Hausdorff distance to each input set.
Let \(\coll=\{A_1, \ldots, A_m\}\) be a collection of \(m\) input shapes with largest pairwise Hausdorff distance~$1$.
We define \(T_\alpha\) as \(\bigcap_i (A_i \oplus D_\alpha)\); the (maximal) middle set is then given by the smallest value~\(\alpha\) for which $T_\alpha \oplus D_\alpha$~contains all input sets.
We denote this smallest \(\alpha\) by $\alpha(\coll):=\min\{\,\alpha \mid \max_i \dh(A_i,T_\alpha) \leq \alpha \,\}$.
If~$\alpha$~is clear from the context,
we use the notation $\dilate{A}$ to mean $A\oplus D_{\alpha}$.

In this section, we first study the largest possible $\alpha(\coll)$ for general and convex input. We then study some general properties of $T_\alpha$ with respect to connectivity and convexity.
After this, we consider whether there is some subset of \(\coll\) that requires the same value of \(\alpha\), and obtain a Helly-type property for convex input.
Finally, we will give various algorithms to compute or approximate $\alpha(\coll)$ efficiently.

\subsection{The largest \texorpdfstring{$\alpha(\coll)$}{a(M)}}
In this section, we are interested in the largest
possible value of $\alpha(\coll)$.
We first discuss the general case and then
study the case where all sets $A\in \coll$ are
convex. In both cases, we provide an exact answer.
This section relies on some tedious calculations, which turn out to be easier if we do not normalize pairwise distances of our objects to 1.

As it turns out, for some inputs it may be the case that \(\alpha(\coll) = 1\); see \cref{fig:three-sets-connected-lower-bound}. Here, there can be no shape with Hausdorff distance less than 1 to all the input shapes, meaning any of the three input shapes can be chosen as ``the middle''.
Hence, for two input sets, we always have $\alpha(\coll)=1/2$, but for more input sets, the value depends on the input, and $\alpha(\coll)$ will be in $[1/2,\,1]$.
The example in \cref{fig:three-sets-connected-lower-bound} requires non-convex sets, raising the question of what the range of $\alpha(\coll)$ can be when all $A_i$ are convex.

If we have three convex sets that are points, and they form the corners of an equilateral unit-side triangle, then we can easily see that $\alpha(\coll)= 1/\sqrt{3}\approx 0.577$ and the middle shape is exactly the point in the middle of the triangle.

An example with three line segments shown
in \cref{fig:three-convex-lower-bound} surprisingly
achieves (for $\lambda\approx0.253135$, $\theta\approx123.37^\circ$) a larger value $\alpha^* \approx 0.6068=r$, which we call the
\emph{\magicValue}.
Lemma~\ref{lem:bound3convex} shows that no three convex sets achieve $\alpha(\coll)>\alpha^*$. Thus the \magicValue is a tight upper bound for three convex sets.

\begin{figure}
    \centering
    \includegraphics{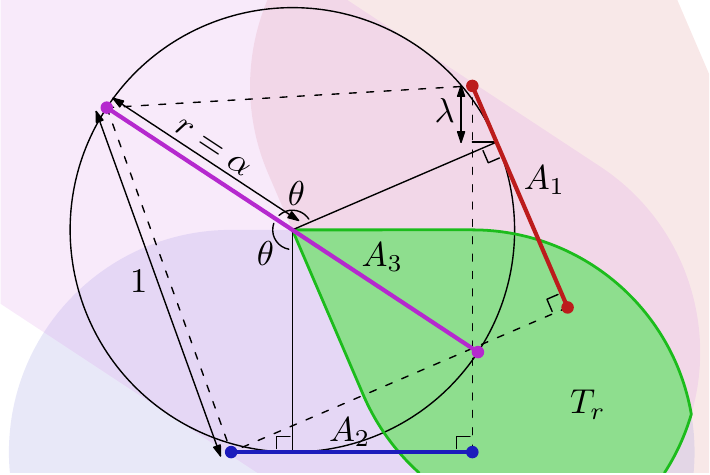}
    \caption{Three segments $A_1$, $A_2$, and $A_3$. Of these, $A_3$ is the diameter of a circle with radius $r$; the other two ($A_1$ and $A_2$) are tangent to the circle and are copies of one another reflected through $A_3$, such that all pairwise Hausdorff distances are at most $1$ (length of dashed segments).
    The top left vertex of $A_3$ is furthest (at distance $r$) from the middle set $T_r$ (green), so $\alpha(\{A_1, A_2, A_3\})$ is the radius $r$ of the circle.}
    \label{fig:three-convex-lower-bound}
\end{figure}

We define the \magicValue as
$\alpha^*=1/z\approx 0.6068$, where the value of $z$ is derived from Figure~\ref{fig:three-convex-trig}, and defined as $z:=\min\{\lambda+1-\cos(2\theta)\mid
\lambda\geq 0\text{, }
\theta\in(90^\circ,180^\circ)\text{, and }
\lambda+1-\cos(2\theta)=\|(-\lambda\cot(2\theta)-\sin(2\theta)+\sin(\theta),\lambda-\cos(2\theta)+\cos(\theta))\|\}\approx 1.647986325231$ (at $\lambda\approx0.253135$, $\theta\approx123.37^\circ$, verified using Wolfram Cloud).

\begin{lemma}\label{lem:convexProperty}
    Let \(\coll=\{A_1,\dots,A_m\}\) be a collection of convex regions in the plane, and $\alpha:=\alpha(\coll)$.
    There is some $A_i \in \coll$ with $\dhd(A_i,T_\alpha)=\alpha$.
\end{lemma}
\begin{proof}
    By construction, we have $\dhd(T_{\beta},A_i)\leq\beta$ for all $i$ and all $\beta$.
    (Recall that this is equivalent to $T_{\beta} \subseteq A_i \oplus D_\beta$.)
    Moreover, if $T_{\beta}$ is nonempty, then for any $i$, the map $\gamma\mapsto\dhd(T_{\gamma},A_i)$ is continuous on the domain $[\beta,\infty)$, as $T_\gamma$ changes continuously.
    We show that for some $i$, we have $\dhd(A_i,T_\alpha)=\alpha$.
    If instead $\dhd(A_i,T_\alpha)<\alpha$ for all $i$, then unless $T_{\beta}$ is empty for all $\beta<\alpha$, we can decrease $\alpha$, contradicting minimality of $\alpha$.
    If instead $\alpha$ is the minimum value for which $T_\alpha$ is nonempty, then either $\alpha=0$ and we are done because $T_\alpha$ contains all $A_i$, or $\alpha>0$ and $T_\alpha$ has no interior (when viewed as a subset of the plane).
    Because $T_\alpha$ is the intersection of convex sets, it is convex.
    If it has no interior, it is either a segment or a point, and by convexity it must lie on the boundary of $\dilate{A_i}$ for some $i$, contradicting that $\dhd(A_i,T_\alpha)<\alpha$.
\end{proof}

\begin{figure}[tbp]
        \centering
        \includegraphics{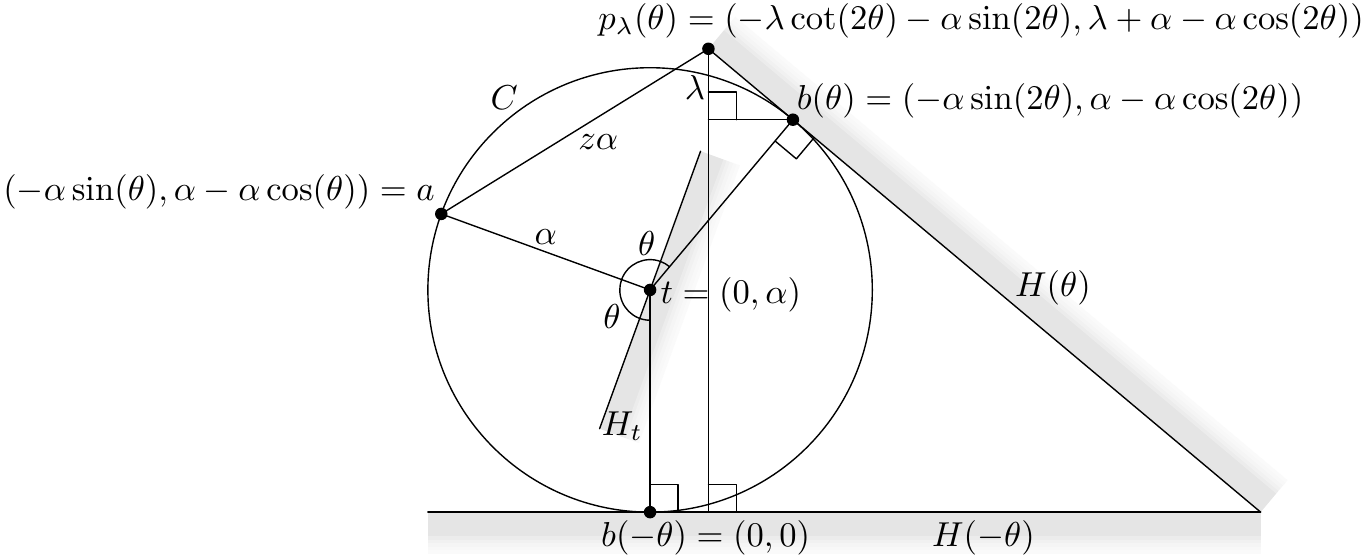}
        \caption{Derivation of the expression for $z$.}
        \label{fig:three-convex-trig}
\end{figure}
\begin{lemma}\label{lem:bound3convex}
    Let \(\coll=\{A_1,A_2,A_3\}\) be convex regions in the plane. Let \(\alpha:=\alpha(\coll)\) and \(d=\max_{i,j}\dh(A_i,A_j)\), then \(d\geq\alpha/\alpha^*\) (equivalently \(d\geq z\alpha\)).
\end{lemma}
\begin{proof}
        By \cref{lem:convexProperty}, we have $\dhd(A_i,T_\alpha)=\alpha$ for some $i$.
        If $x$ is a point, we will write $\dd(x,\cdot)$ to denote $\dhd(\{x\},\cdot)$.
        Without loss of generality assume that $\dhd(A_3,T_\alpha)=\alpha$ and $\dd(a,T_\alpha)=d(a,t)=\alpha$ with $a\in A_3$ and $t\in T_\alpha$.
        Let $T=\dilate{A_1}\cap\dilate{A_2}\supseteq T_\alpha$.
        There is no point $t'\in T$ with $d(t',a)<\alpha$, since then $\dd(t',A_3)<\alpha$, in which case $t'\in \dilate{A_3}$ and therefore $t'\in T_\alpha$, contradicting that $\dhd(a,T_\alpha)=\alpha$.
        So $t$ is a point in $T$ closest to $a$ and hence $\dd(a,T)\geq\alpha$.

        Assume that $\alpha>0$ (otherwise we are done) and let $H_t$ be the half-plane (not containing~$a$) bounded by the line through $t$ that is perpendicular to $\seg(t,a)$, see also \cref{fig:three-convex-trig}.
        The set $T$ is convex, as it is the intersection of convex sets.
        Therefore, if $T$ contains a point $p$, then $T$ also contains $\seg(t,p)$.
        Since $t$ is a point of $T$ closest to $a$, no such segment intersects the open disk of radius $\alpha$ centered at $a$, and therefore $T\subseteq H_t$.

        Let $C$ be the circle of radius $\alpha$ centered at $t$.
        For the remainder of the proof, let $i\in\{1,2\}$.
        Let $b_i$ be a point of $A_i$ closest to $t$.
        Then $b_i$ lies on or inside $C$.
        If $b_i\neq t$, we can define the half-plane $H_i$ (not containing $t$) bounded by the line through $b_i$ that is perpendicular to $\seg(t,b_i)$.
        For $b_i\neq t$, we have by convexity of $A_i$ and $b_i$ being closest to $t$ that $A_i\subseteq H_i$, so $\dd(a,A_i)\geq\dd(a,H_i)$.
        Without loss of generality, assume that $\dd(a,H_i)<\alpha/\alpha^*$ (otherwise $d\geq\dhd(A_3,A_i)\geq\dd(a,A_i)\geq\alpha/\alpha^*$).

        If $d(a,b_i)\geq 2\alpha$, then $b_i$ lies diametrically opposite to $a$ on $C$, but then $\dd(a,H_i)\geq 2\alpha>\alpha/\alpha^*$, which is a contradiction, so $d(a,b_i)<2\alpha$.
        Let $t_i\in\dilate{A_i}$ be the midpoint of $b_i$ and $a$, then $d(a,t_i)<\alpha=d(a,t)$.
        If $d(b_1,t)<\alpha$, then $T$ contains a point interior to $\seg(t,t_2)$, contradicting that $\dd(a,T)\geq\alpha$.
        So $b_1$ and (analogously) $b_2$ lie on $C$.

        Let~$\theta_i$ be the clockwise angle $\measuredangle atb_i\in(-180^\circ,180^\circ)$.
        Define $b(\theta)$ to be the point on $C$ for which $\theta$ is the clockwise angle $\measuredangle atb(\theta)$, so that $b_i=b(\theta_i)$.
        Similarly, let $H(\theta)$ be the half-plane (not containing $C$) bounded by the line tangent to $C$ at $b(\theta)$, so that $H_i=H(\theta_i)$.
        Assume without loss of generality that $|\theta_1|\geq|\theta_2|$ (otherwise relabel $A_1$ and $A_2$).
        If $\theta_1$ and $\theta_2$ are both positive or both negative, consider the circle of radius $\alpha/2$ centered at the midpoint of $a$ and $t$. Then $t_1$ lies on the (shorter) arc of this circle connecting $t_2$ and $t$.
        This arc lies entirely in $\dilate{A_2}$, so $t_1$ lies in $T_\alpha$, which contradicts that there is no point $t'\in T$ with $d(t',a)<\alpha$.
        So assume without loss of generality that $\theta_2\leq 0\leq\theta_1$ (otherwise mirror all points).
        If $\theta_1-\theta_2<180^\circ$, then $T$ contains the segment between $t$ and the midpoint of $b_1$ and $b_2$.
        This segment does not lie in $H_t$, which contradicts that $T\subseteq H_t$.
        Moreover, if $\theta_1-\theta_2=180^\circ$, then $b_1$ and $b_2$ are antipodal on $C$, so $\dh(A_1,A_2)\geq\dh(H(\theta_1),H(\theta_2))=2\alpha>\alpha/\alpha^*$.
        So consider the remaining case where $\theta_1-\theta_2>180^\circ$.

        In fact, it will turn out that in the worst case, $\theta_2=-\theta_1$.
        Suppose that $p\in A_1\subseteq H(\theta_1)$ is the point of $A_1$ closest to $a$.
        We have $d\geq d(a,p)$ and $d\geq\dd(p,A_2)\geq\dd(p,H(\theta_2))$.
        Moreover, since $-\theta_1\leq\theta_2<\theta_1-180^\circ$, the value of $\dd(p,H(\theta))$ decreases as $\theta\in[-\theta_1,\theta_2]$ decreases.
        In particular, we have $\dd(p,H(\theta_2))\geq\dd(p,H(-\theta_1))$.
        Since $|\theta_1|\geq|\theta_2|$, we have $\theta_1\in(90^\circ,180^\circ)$.
        Let $\lambda_p=\dd(p,H(-\theta_1))-\dd(b(\theta_1),H(-\theta_1))$.
        If $\lambda_p<0$, then $d(a,b(\theta_1))<d(a,p)$, and $p$ would not be a point of $A_1$ closest to $a$ because the angle $\measuredangle ab(\theta_1)p$ would be at least 90 degrees.
        Combining the above lower bounds, we obtain $d\geq\min\{\max\{d(a,p),\dd(p,H(-\theta_1))\}\mid p\in H(\theta_1),\lambda_p\geq 0\}$.
        The right hand side of the above inequality is attained for some $p$ on the boundary of $H(\theta_1)$.
        We parameterize such points $p$ with parameters $\lambda$ and $\theta_1$: let $p_\lambda(\theta_1)$ be the unique point on the boundary of $H(\theta_1)$ with $\dd(p_\lambda(\theta_1),H(-\theta_1))=\dd(b(\theta_1),H(-\theta_1))+\lambda$.
        The above inequality becomes $d\geq\min_{\lambda\geq 0}\max\{d(a,p_\lambda(\theta_1)),\dd(b(\theta_1),H(-\theta_1))+\lambda\}$.

        We need to minimize this quantity over all values of $\lambda\geq 0$ and $\theta_1\in(90^\circ,180^\circ)$.
        We will show that it is minimized when its terms $d(a,p_\lambda(\theta_1))$ and $\dd(b(\theta_1),H(-\theta_1))+\lambda$ are equal.
        The point $p_\lambda(\theta_1)$, and hence the two terms, vary continuously in $\lambda$ and $\theta_1$.
        For fixed $\theta_1$, both terms are convex as a function of $\lambda$.
        Therefore, for any fixed $\theta_1$, the function is minimized either when $\lambda=0$, or the two terms are equal.
        As $\theta_1$ approaches $180^\circ$, the first term approaches at least $2\alpha$ (for any $\lambda$), and as $\theta_1$ approaches $90^\circ$, the second term approaches at least $2\alpha$.
        Since the optimal value is less than $2\alpha$, there exists an optimal value of $\theta_1$.
        Assume for a contradiction that the terms are not equal in an optimal solution.
        Fix $\lambda=0$, and consider the two terms as a function of $\theta_1$.
        For $\theta_1\approx 90^\circ$ and $\lambda=0$, we have $d(a,p_\lambda(\theta_1))\approx\alpha<2\alpha\approx \dd(b(\theta_1),H(-\theta_1))+\lambda$.
        Conversely for $\theta_1\approx 180^\circ$ and $\lambda=0$, we have $d(a,p_\lambda(\theta_1))\approx 2\alpha>0\approx \dd(b(\theta_1),H(-\theta_1))+\lambda$.
        Hence, by the intermediate value theorem, the inequality as a function of $\theta_1$ (with fixed $\lambda=0$) is minimized when the terms are equal.
        We handled the case with $\lambda>0$ above, so our inequality becomes $d\geq\min\{\dd(b(\theta_1),H(-\theta_1))+\lambda\mid \lambda\geq 0, \theta_1\in(90^\circ,180^\circ), \text{ and } d(a,p_\lambda(\theta_1))=\dd(b(\theta_1),H(-\theta_1))+\lambda \}$.
        Following the derivation in \cref{fig:three-convex-trig}, this corresponds to~$d\geq z\alpha=\alpha/\alpha^*$.
    \end{proof}

\subsection{Convexity and connectedness of \texorpdfstring{\(\boldsymbol{T_\alpha}\)}{S-alpha}}
In this subsection, we use \(\alpha := \alpha(\coll)\) for simplicity. Similar to \cref{sec:two-sets-properties}, we examine the properties of \(T_\alpha\) for different types of input. We arrive at straightforward generalizations of the results obtained for two sets.

\begin{enumerate}
    \item If all \(A_i\) are convex, then \(T_\alpha\) is convex.
    \item If one of the \(A_i\) is connected and the rest are convex, then \(T_\alpha\) is connected.
    \item For some input where each \(A_i\) is connected, and at least two are not convex, \(T_\alpha\) is disconnected.
\end{enumerate}

Property 1 follows from the same argument as before: \(T_\alpha\) is the intersection of convex sets, and therefore itself convex. Property 3 can be shown by extending the construction from \cref{fig:two-sets-quadratic-lower-bound} with some other sets: if the intersection of two of the sets is not connected, adding more sets will not make \(T_\alpha\) connected as long as the pairwise Hausdorff distance does not increase. We establish Property 2 with the following lemma.

\begin{lemma}
    Let \(\coll = \{A_1, \ldots, A_m\}\) be a set of connected regions of the plane, with \(A_i\) convex for \(i < m\). Then \(T_\alpha\) is connected.
\end{lemma}
\begin{proof}
    Consider the set \(T_\alpha' = \bigcap_{i = 1}^{m - 1} \dilate{A_i}\). This set is convex, as it is the intersection of convex sets. Also note that by definition of \(T_\alpha\), \(A_m\) has directed Hausdorff distance at most \(\alpha\) to \(T_\alpha'\). Let \(A = T_\alpha'\) and \(B = A_m\), normalised such that \(\dhd(B, A) = 1\). We now apply \cref{lem:one-non-convex-connected} to \(A\) and \(B\), using zero as the value for \(\alpha\). We obtain the result that \(T_\alpha = T_\alpha' \oplus D_0 \cap A_m \oplus D_\alpha\) is connected. Note that the Hausdorff distance from \(A\) to \(B\) may be bigger than one, but this does not matter for the proof of \cref{lem:one-non-convex-connected}.
\end{proof}

\begin{figure}
    \begin{subfigure}{0.3\textwidth}
        \includegraphics[width=\textwidth,page=1]{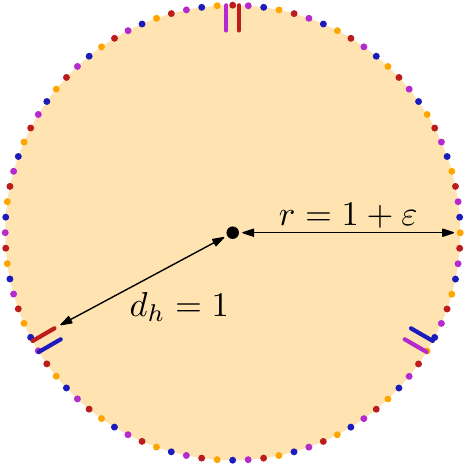}
        \caption{}
        \label{fig:non-convex-no-helly-a}
    \end{subfigure}
    \hfill
    \begin{subfigure}{0.3\textwidth}
        \includegraphics[width=\textwidth,page=2]{figures/non-convex-no-helly.pdf}
        \caption{}
        \label{fig:non-convex-no-helly-b}
    \end{subfigure}
    \hfill
    \begin{subfigure}{0.3\textwidth}
        \includegraphics[width=\textwidth,page=4]{figures/non-convex-no-helly.pdf}
        \caption{}
        \label{fig:non-convex-no-helly-c}
    \end{subfigure}
    \caption{When the input sets are not convex, all sets may be necessary to realise the value of \(\alpha\). Figure \subref{fig:non-convex-no-helly-a} shows our input construction, along with the radius of the circle and the Hausdorff distance. Figure \subref{fig:non-convex-no-helly-b} shows that when all sets are present, the required value of \(\alpha\) is \((1 + \eps) / 2\). Figure \subref{fig:non-convex-no-helly-c} shows that with the red set removed, the required value of \(\alpha\) is reduced to \((1 + \eps / 2) / 2\).}
    \label{fig:non-convex-no-helly}
\end{figure}

\subsection{Helly-type properties}

An interesting question is whether there are any sets in the input
that could be removed while maintaining the same optimal value of \(\alpha\).
To make this precise, we need some definitions.
We say  a collection $\coll$ of $m$ sets  is \emph{$d$-sufficient},
if there is a collection $\coll_d \subset \coll$ of $d$ sets
such that $\alpha(\coll) = \alpha(\coll_d)$.

\begin{lemma}
    For every $m$, there is a collection $\coll$ of $m$ connected sets in the plane that
    is not $(m-1)$-sufficient.
\end{lemma}
\begin{proof}
\cref{fig:non-convex-no-helly} depicts a construction of four sets which are not $3$-sufficient, which generalises to more sets.
The example has one set that is a disk of radius \(1 + \eps\) (shown in orange in \cref{fig:non-convex-no-helly-a}), and \(m - 1\) sets that are circles on the boundary of this disk with $m-1$ protrusions
of some small length \(\eps\).
These protrusions are evenly spaced along the boundary of the disk, and in each location there is a distinct set out of the \(m - 1\) sets missing (each subset of size \(m - 2\) is represented by some protrusion). In the example, we have protrusions containing the red and the blue set, the blue and the purple set, and the red and the purple set.
This way, for the case where all sets are present (\cref{fig:non-convex-no-helly-b}), the protrusions don't have any influence on \(T_\alpha\), meaning that \(\alpha \geq (1 + \eps) / 2\)
is required to let \(\dilate{T_\alpha}\) contain the entire disk.
However, if we remove one set (other than the orange disk), there will be one protrusion where all sets are now present, meaning it will change the shape of \(T_\alpha\). Because of this, the center of the disk will already be covered with a smaller value of \(\alpha\), namely \((1 + \eps / 2) / 2\). This is shown in \cref{fig:non-convex-no-helly-c}: the dotted arc shows the dilation of the bump caused by the protrusion, which covers a part of the disk that would otherwise not be covered (shown as a dashed circle).
Note that if we remove the orange disk, it is sufficient to use a value of \(\alpha = \eps / 2\).
Further note that with a minor adaptation, all sets become polygonal and simply connected.
\end{proof}

We have shown that in general, we cannot remove any sets from the input while maintaining the same value of \(\alpha\). However, when all input sets are convex, we can show that there is always a subset of size at most three that has the same optimal value of \(\alpha\).

\begin{lemma}\label{lem:helly}
    Let \(\coll = \{A_1, \ldots, A_m\}\) be a collection of convex sets.
    Then there exists a subcollection \(\coll' \subseteq \coll\) of size at most three such that \(\alpha(\coll) = \alpha(\coll')\).
\end{lemma}
\begin{proof}
    Consider growing some value \(\beta\) from \(1/2\) to $1$.
    At some point, \(\dilate{T_\beta}\)
    contains all sets in~\(\coll\).
    There are two ways in which this can happen: (1) \(T_\beta\) is non-empty for the first time, and immediately the condition holds, or (2) \(T_\beta\) grows, and its dilation now covers the last point of all sets in \(\coll\).
    As $T_\beta$ is convex no new components can appear except for the first, and thus we have only those two cases.

    In Case~1, \(T_\beta\) is either a segment or a point; otherwise, \(T_{\beta'}\) would have been non-empty for some \(\beta' < \beta\). If it is a segment, it is generated by two parallel edges of some \(A_i,A_j \in \coll\)
    such that we have
    $\alpha(\{A_i,A_j\}) = \alpha(\coll)$.
    If it is a point, it is the common intersection of the dilation of some number of sets from \(\coll\); we argue that you can always pick three sets for which \(\beta\) is optimal.
    Let \(a\) be the single point
    in \(T_\beta\); consider the vectors $\mathcal{V}$ perpendicular to the boundaries of the dilated input sets intersecting in this point. The vectors $\mathcal{V}$ must
    positively span the plane:\footnote{We say $v_i \in \R^2$ \emph{span the plane positively}, if for every point $p \in \R^2$ there are some numbers $a_i \in \R^+$ such that $\sum a_i v_i  = p$.} otherwise, all vectors would lie in a common half-plane, and \(a\) would not be the first point to appear in \(T_\beta\).
    As we are in the plane, there must be subset $\mathcal{U} \subset \mathcal{V} $ of three vectors that positively span the plane by themselves. The three corresponding sets $A_i,A_j,A_k \in \coll$
    satisfy
    $\alpha(\{A_i,A_j,A_k\}) = \alpha(\coll)$.

    In Case 2, as our input sets are convex, \(T_\beta\) itself is also convex. Let \(a \in A_i\) be one of the last points of \(\coll\) to be covered by \(T_\beta^\oplus\). As \(T_\beta^\oplus\) is convex, \(a\) must be on its boundary; let \(c\) be the piece of boundary curve \(a\) lies on. This piece of curve is either generated by the dilation of some boundary curve in \(T_\beta\), or by the dilation of one of its vertices. If it's the dilation of a boundary curve, it can be traced back directly to boundary curve of some \(A_j\), in which case \(A_i\) and \(A_j\) have Hausdorff distance \(2\beta\), and \(\alpha(\{A_i,A_j\}) = \alpha(\coll)\) for any choice of \(k\). If it can be traced back to a vertex of \(T_\beta\), this vertex is generated by the intersection of the boundaries of some \(A_j^\oplus\) and \(A_k^\oplus\), in which case we also have that \(\alpha(\{A_i,A_j,A_k\}) = \alpha(\coll)\).
\end{proof}

Combining the previous lemma with \cref{lem:bound3convex}, we obtain the following result.

\begin{theorem}
Let \(\coll = \{A_1, \ldots, A_m\}\) be a collection of convex regions in the plane, and let \(T_\alpha = \bigcap_i \dilate{A_i}\). Then \(\alpha(\coll)\) is at most the \magicValue \(\alpha^* \approx 0.6068\).
\end{theorem}

\subsection{Algorithms}
\label{ss:algo}

For any given collection of polygons $\coll = \{A_1,\ldots,A_m\}$, we want to compute $\alpha(\coll)$.
We present two algorithms, a simple approximation algorithm and a more complex exact algorithm. They both use the same decision algorithm as a subroutine.
To be precise, given a collection of sets $\coll$ and some $\alpha$, the decision algorithm decides if $\alpha \leq \alpha(\coll)$.
We first present an algorithm for the decision problem.
Then we sketch how they are used in the approximation algorithm and the exact algorithm. We denote all vertices and edges of the $A_i$ as \emph{features} of $\coll$.

\paragraph{Decision algorithm}
Assuming the input has total complexity \(n\), we can test whether a given value of $\alpha \leq \alpha(\mathcal{M})$ as follows. Compute the intersection $T_\alpha$ of the dilations $\dilate{A_1},\ldots,\dilate{A_m}$ in \(O(n^2\log n)\) time, using the construction of an arrangement of straight and circular arcs~\cite{edelsbrunner1992arrangements,halperin18}.
The set \(T_\alpha\) will always have at most quadratic
complexity, but it can be disconnected.
Next we compute $\dilate{T_\alpha}$. We take every connected component $T$ of $T_\alpha$ separately, compute $\dilate{T}$,
and then compute their union. Since the connected components of $T_\alpha$ are
disjoint and can be partitioned into $O(n^2)$ convex pieces,
the Minkowski sums of these pieces with $D_\alpha$ form a set of pseudo-disks with total complexity $O(n^2)$, see~\cite{kedem1986union}.
It is known that such a union has $O(n^2)$ complexity and can be computed in $O(n^2\log^2 n)$ time~\cite{agarwal2008state,kedem1986union}.
Thus, we can compute $T_\alpha$ in $O(n^2\log^2 n)$ time.

Note that $T_\alpha \subseteq \dilate{A_i}$, by definition. It remains to test $A_i \subseteq \dilate{T_\alpha}$, for each $A_i$.
We test all those containments
by a standard plane sweep~\cite{bcko-cgaa-08} in \(O(n^2 \log n)\) time.
As soon as we find any proper intersection between an arc of $\partial (\dilate{T_\alpha})$ and some edge of some $\partial A_i$, we can stop the sweep and conclude that $\alpha$ needs to be larger. If there were no proper intersections of this type, there were only $O(n^2)$ events (and not $O(n^3)$), including the ones between edges of different $\partial A_i$.
When there are no proper intersections, each shape $A_i$ lies fully inside or outside $\dilate{T_\alpha}$. We can test this in $O(n^2\log n)$ time (replace each $A_i$ by a single point and then test by a plane sweep or planar point location~\cite{bcko-cgaa-08}), and conclude that $\alpha$ must be larger or smaller than the one tested.
Thus this decision algorithm takes $O(n^2 \log^2 n)$ total time.

\paragraph{Approximation algorithm}
The decision algorithm leads to a simple
approximation algorithm to find a value of \(\alpha\) that is at most a factor \(1 + \eps\) from the optimum.
We can perform \(\lceil\log 1/\eps\rceil\) steps of binary search in the range \([1/2, 1]\), testing if \(\dilate{T_\alpha}\) contains all \(A_i\) using the above decision algorithm.
This takes \(O(n^2 \log^2 n\log 1 / \eps)\) time in total.

\begin{figure}[bt]
\centering
    \includegraphics{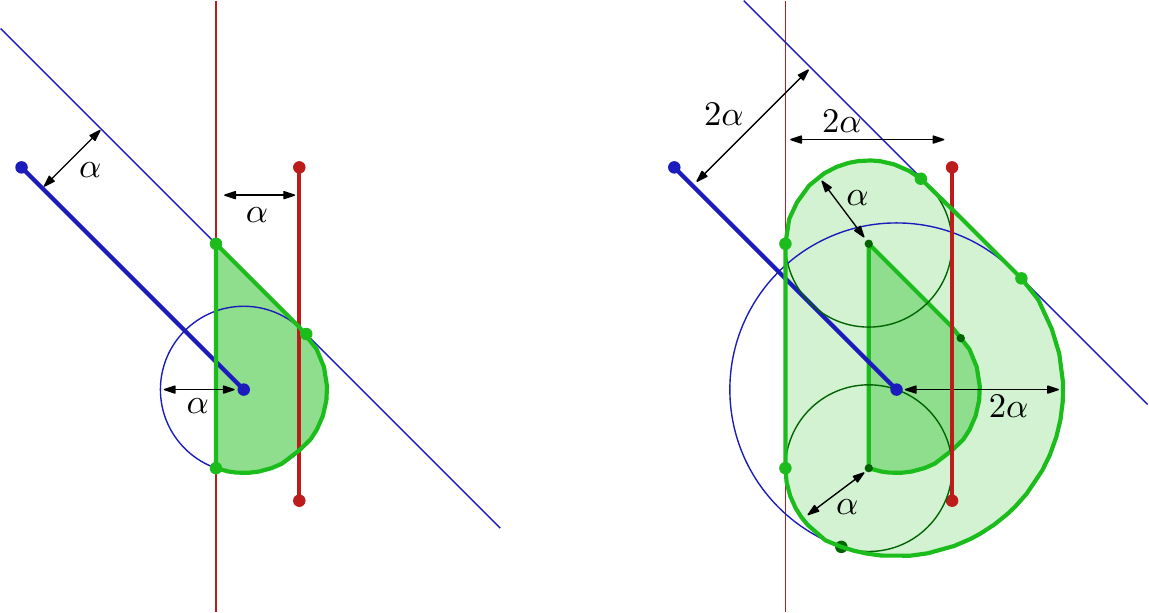}
    \caption{Left, two sets shown by red and blue line segments, and the construction of $T_\alpha$ from lines parallel to edges of $\coll$ and circles centered at vertices of $\coll$.
    Right, construction of $\dilate{T_\alpha}$ from lines at distance $2\alpha$ from edges of $\coll$, circles of radius $2\alpha$ centered at vertices of $\coll$, and circles of radius $\alpha$ centered at certain vertices of $T_\alpha$.}
    \label{fig:arrangements}
\end{figure}

\paragraph{Exact computation}
We can compute an exact value of $\alpha(\coll)$ in polynomial time. To this end, we imagine a continuous process where we grow $\alpha$ from $1/2$, and keep track of $\dilate{T_\alpha}$. The first time (smallest $\alpha$) $\dilate{T_\alpha}$ covers all $A_i$, we have found the Hausdorff distance $\alpha(\coll)$ corresponding to the Hausdorff middle, and we can construct $T_\alpha$ explicitly as the Hausdorff middle.
Such an approach is sometimes called \emph{wavefront propagation} or \emph{continuous Dijkstra}; it has been used before to compute Voronoi diagrams~\cite{bcko-cgaa-08,fortune1987sweepline},
straight skeletons~\cite{aichholzer1995novel} and shortest paths on terrains~\cite{mitchell1987discrete}. This approach is combinatorial if there are finitely many events and we can determine each on time, before it occurs.
Instead of explicitly maintaining $\dilate{T_\alpha}$ when $\alpha$ grows, we will determine a polynomial-size set of critical $\alpha$ values that contains the sought one, and find it by binary search,
using the decision algorithm described above.

The value $\alpha(\coll)$ that we aim to compute occurs when $\dilate{T_\alpha}$ has grown just enough to cover all $A_i$. This can happen in three ways, roughly corresponding to a vertex of $A_i$ becoming covered, an edge of $A_i$ becoming covered at some point ``in the middle'', or a hole of $\dilate{T_\alpha}$ collapsing and disappearing interior to $A_i$.
We call the vertices, edges, and arcs of $\coll$ and
$\dilate{T_\alpha}$ the \emph{features} (of their boundaries).
The three ways of covering all $A_i$, expressed in the features of $\coll$ and $\dilate{T_\alpha}$, are now: (1) a feature of $\dilate{T_\alpha}$ coincides with a vertex of some $A_i$, (2) a vertex of $\dilate{T_\alpha}$ lies on a feature of some $A_i$, or (3) features of $\dilate{T_\alpha}$ collapse and cause a hole of $\dilate{T_\alpha}$ to disappear.
In the last case, when that hole was inside some $A_i$, this can be the event where $A_i$ is covered fully for the first time. In all cases, one, two, or three features of $\dilate{T_\alpha}$ and zero or one  feature of some $A_i$ are involved, and at
most three features in total. When three edge or circular arc features pass through a single point for some value of $\alpha$, we say that these features are \emph{concurrent}. Similarly, when an edge or circular arc passes through a vertex for some $\alpha$, we say they are concurrent.

It can be that more than three features of $\dilate{T_\alpha}$ pass through the point where, e.g., a hole in $\dilate{T_\alpha}$ disappears, but then we can still determine this critical value by examining just three features of $\dilate{T_\alpha}$, and computing the $\alpha$ value when the curves of these three features are concurrent.

Let us analyze which features make up the boundary of $\dilate{T_\alpha}$, see Figure~\ref{fig:arrangements}.
There are four types: (1) straight edges, which are at distance $2\alpha$ from an edge of $\coll$, and parallel to it, (2) circular arcs of radius $2\alpha$, which are parts of circles centered at vertices of $\coll$, (3) circular arcs of radius $\alpha$, centered at a vertex of $T_\alpha$, and (4) vertices where features of types (1)--(3) meet. Every one of the features of the boundary of $\dilate{T_\alpha}$ is determined by one or two features of $\coll$. In particular, each arc of type (3) is centered on
an intersection point which is a vertex of $T_\alpha$, of which there can be $\Theta(n^2)$ in the worst case (Figure~\ref{fig:two-sets-quadratic-lower-bound}).
Depending on the type of intersection point, its trace may be linear in $\alpha$, or may follow a low-degree algebraic curve (when the intersection has equal distance $\alpha$ to an edge and a vertex of $\coll$).

Since any critical value can be determined as a concurrency of two (vertex and edge or arc) or three features (three edges or arcs) from $\coll$ and $\dilate{T_\alpha}$, and features of $\dilate{T_\alpha}$ in turn are determined by up to two features of $\coll$, every critical value
depends on at most six features of the input $\coll$.
If we choose any tuple with up to six features of
$\coll$, and compute the $\alpha$ values that may be critical, we obtain a set of $O(n^6)$ values that contain all critical $\alpha$ values, among which $\alpha(\coll)$. We can compute this set in $O(n^6)$ time, as it requires $O(1)$ time for each tuple of up to six features of $\coll$.

\begin{theorem}
Let $\coll$
be a collection of $m$ polygonal shapes in the plane with total complexity~$n$, such that the Hausdorff distance between any pair is at most $1$, and let $\eps>0$ be a constant.
The Hausdorff middle of \(\mathcal{M}\) can be computed exactly in $O(n^6)$ time, and approximated within $\eps$
in $O(n^2\log^2  n \log 1/\eps)$ time.
\end{theorem}

Parametric search could result in a faster exact algorithm, but for this one would need to express whether input features are close to a given \(S_\alpha\) in terms of low-degree polynomials. This is nontrivial given that \(S_\alpha\) as function of \(\alpha\) varies in a complex manner.

\section{Discussion and future research}

We have defined and studied the Hausdorff middle of two planar sets, leading to a new morph between these sets. We also considered the Hausdorff middle for more than two sets.
While we assumed that the input sets are simply connected, our definition of middle and the morph immediately generalize to more general sets, like sets with multiple components and holes. In this sense our definition of middle is very general. Other interpolation methods between shapes do not generalize to more than two input sets and cannot easily handle sets with multiple components.

There are many interesting open questions.
For example, when both input sets are one-dimensional curves, is there a natural way to define a Hausdorff middle curve that is also $1$-dimensional?

Besides the maximal middle set, there are other options for a Hausdorff middle. For example, we can choose $S_\alpha$ clipped to the convex hull of $A\cup B$,
which is also a valid Hausdorff middle. In Figure~\ref{fig:cost-of-connectedness}, the green shape would be reduced to the part inside the square, which may be more natural. This Hausdorff middle can also be used in a morph.

Another interesting question could be if, for two shapes \(A\) and \(B\), we can find a translation or rigid motion of \(A\) such that some measure on the Hausdorff middle (e.g. area, perimeter, diameter) is minimised.

For two or more shapes in the plane, we could also define a middle based on area of symmetric difference. Here we may want to average the areas for the middle shape, and possibly choose the middle that minimizes perimeter. This problem is related to minimum-length area bisection~\cite{koutsoupias1992optimal}.

Similarly, for a set of curves, we could define
a middle curve based on the Fr\'echet distance. This appears related to the Fr\'echet distance of a set of curves rather than just a pair~\cite{dumitrescu2004frechet}.

\bibliographystyle{plain}
\bibliography{bibliography}

\end{document}